\documentclass[11pt]{article}

\usepackage[letterpaper,margin=1.00in]{geometry}
\usepackage{amsmath, amssymb, amsthm, thmtools, amsfonts}
\usepackage{bbm}
\usepackage{cite}
\usepackage{appendix}
\usepackage{graphicx}
\usepackage{color}
\usepackage{algorithm}
\usepackage[noend]{algpseudocode}
\usepackage{xspace}
\usepackage{mathtools}
\usepackage{tikz}
\usepackage{pgfplots}
\pgfplotsset{width=7.6cm,compat=1.9} 
\usetikzlibrary{arrows}
\usetikzlibrary{decorations.pathreplacing}
\newtheorem{theorem}{Theorem}[section]
\newtheorem{lemma}[theorem]{Lemma}
\newtheorem{meta-theorem}[theorem]{Meta-Theorem}

\newtheorem{corollary}[theorem]{Corollary}

\newtheorem{definition}[theorem]{Definition}

\definecolor{darkgreen}{rgb}{0,0.5,0}
\definecolor{darkblue}{rgb}{0,0,0.5}
\usepackage{hyperref}
\hypersetup{
    unicode=false,
    colorlinks=true,
    linkcolor=darkblue,
    citecolor=darkgreen,
    filecolor=magenta,
    urlcolor=cyan
}
\usepackage[capitalize, nameinlink]{cleveref}
\crefname{theorem}{Theorem}{Theorems}
\Crefname{lemma}{Lemma}{Lemmas}
\Crefname{observation}{Observation}{Observations}
\Crefname{remark}{Remark}{Remarks}
\Crefname{equation}{}{}
\renewcommand{\paragraph}[1]{\vspace{0.15cm}\noindent {\bf #1}:}

\newcommand{\pushcode}[1][1]{\hskip\dimexpr#1\algorithmicindent\relax}

\newcommand{\skueue}[0]{\textsc{Skueue}}
\newcommand{\pname}[0]{\textsc{Skeap}}
\newcommand{\pnamep}[0]{\textsc{Seap}}
\newcommand{\kselect}[0]{\textsc{KSelect}}
\newcommand{\ins}[1]{\textsc{Insert}\ensuremath{(#1)}\xspace}
\newcommand{\delmin}{\textsc{DeleteMin}\ensuremath{()}\xspace}
\newcommand{\join}[1]{\textsc{Join}\ensuremath{(#1)}\xspace}
\newcommand{\leave}{\textsc{Leave}\ensuremath{()}\xspace}

\title{Skeap \& Seap: Scalable Distributed Priority Queues for Constant and Arbitrary Priorities\footnote{This is an extended version of a paper which will appear in SPAA 2019. This work has been partially supported by the German Research Foundation (DFG) within the Collaborative Research Center 901 ''On-The-Fly Computing'' under the project number 160364472-SFB901.
}}
\author{
  Michael Feldmann\\
  \small Paderborn University \\
  \small michael.feldmann@upb.de\\
  \and
  Christian Scheideler\\
  \small Paderborn University \\
  \small scheideler@upb.de\\
}

\begin{document}

\maketitle

\begin{abstract}
	We propose two protocols for distributed priority queues (for simplicity denoted \textit{heap}) called \pname{} and \pnamep{}.
	\pname{} realizes a distributed heap for a constant amount of priorities and \pnamep{} one for an arbitrary amount.
	Both protocols build on an overlay, which induces an aggregation tree on top of which heap operations are aggregated in batches, ensuring that our protocols scale even for a high rate of incoming requests.
	As part of \pnamep{} we provide a novel distributed protocol for the $k$-selection problem that runs in $O(\log n)$ rounds w.h.p.
	\pname{} guarantees sequential consistency for its heap operations, while \pnamep{} guarantees serializability.
	\pname{} and \pnamep{} provide logarithmic runtimes w.h.p. on all their operations with \pnamep{} having to use only $O(\log n)$ bit messages.
\end{abstract}

\section{Introduction} \label{sec:introduction}
Data structures play an important role for both sequential and distributed applications.
While several types of sequential data structures are already well-studied, purely distributed data structures exist mostly in form of distributed hash tables (DHTs).
Other types like queues, stacks or heaps are being considered much less, even though they have some interesting applications as well.
Recently, a distributed queue has been presented~\cite{DBLP:conf/ipps/FeldmannSS18} which can also be extended to a distributed stack~\cite{DBLP:journals/corr/abs-1802-07504}.
Such a queue can be used, for instance, to realize fair work stealing in scheduling, for distributed mutual exclusion or for distributed counting.
This paper is the first to provide a distributed solution for another (basic) data structure, the distributed priority queue (for simplicity denoted \textit{heap}).
A distributed heap may be useful, for example in scheduling, where one may insert jobs that have been assigned priorities and workers may pull these jobs from the heap based on their priority.
Another application for a distributed heap is distributed sorting.

We consider processes that are able to flexibly interconnect using an overlay network.
There are various challenges to overcome when constructing a distributed data structure, such as scalability (for both, the number of participating processes as well as the rate of incoming requests), semantical correctness and distribution of elements among processes.
The specific challenge is to deal with priorities that are assigned to heap elements.
We distinguish between settings which only allow a constant amount of priorities and settings for arbitrary amounts and present two novel distributed protocols \pname{} and \pnamep{} for these scenarios.
Both protocols support insertions and deletions of elements in time $O(\log n)$ w.h.p., where $n$ is the number of processes participating in the heap.
Furthermore, we provide some guarantees on the semantics, by having \pname{} guarantee sequential consistency and \pnamep{} guarantee serializability.
As part of \pnamep{} we obtain a novel protocol \textsc{KSelect} for distributed $k$-selection that runs in $O(\log n)$ rounds w.h.p.
Both \pname{} and \pnamep{} work in the asynchronous message passing model.
As an additional feature we can handle join and leave requests of processes in time $O(\log n)$ w.h.p. without violating the heap semantics or losing important data.
Even though \pnamep{} comes with slightly weaker semantics than \pname{}, it only uses $O(\log n)$ bit messages for its operations, while the message size in \pname{} partially depends on the rate with which processes generate new operations.

\subsection{Model}
We study distributed heaps consisting of multiple processes that are interconnected by some overlay network.
We model the overlay network as a directed graph $G = (V, E)$, where each node $v \in V$ represents a single process, $n = |V|$ and an edge $(v,w)$ indicates that $v$ knows $w$ and can therefore send messages to $w$.
Each node $v$ can be identified by a unique identifier $v.id \in \mathbb{N}$.
There is no global (shared) memory, so only the local memory of the nodes can be used to store heap elements.
We allow the storage capacity of each node to be at most polynomial in $n$.

We consider the asynchronous message passing model where every node $v$ has a set $v.Ch$ for all incoming messages called its \emph{channel}. That is, if a node $u$ sends a message $m$ to node $v$, then $m$ is put into $v.Ch$. A channel can hold an arbitrary finite number of messages and messages never get duplicated or lost.

Nodes may execute \textit{actions}: An action is just a standard procedure that consists of a name, a (possibly empty) set of parameters, and a sequence of statements that are executed when calling that action.
It may be called locally or remotely, i.e., every message that is sent to a node contains the name and the parameters of the action to be called. We will only consider messages that are remote action calls. An action in a node $v$ is {\em enabled} if there is a request for calling it in $v.Ch$. Once the request is processed, it is removed from $v.Ch$. 
We assume \emph{fair message receipt}, i.e., every request in a channel is eventually processed.
Additionally, a node may be \emph{activated} periodically.
Activation of a node is considered to trigger an (specific) action.
Upon activation of node $v$, $v$ may generate messages based on its local information if some protocol-specific conditions are satisfied.

We define the \textit{system state} to consist of the values of all protocol-specific variables of the nodes and the set of messages in each channel. A \textit{computation} is a potentially infinite sequence of system states, where the state $s_{i+1}$ can be reached from its previous state $s_i$ by executing an action that is enabled in $s_i$.

We place no bounds on the message propagation delay or the relative node execution speed, i.e., we allow fully asynchronous computations and non-FIFO message delivery.

For the performance analysis only, we assume the standard synchronous message passing model, where time proceeds in \emph{rounds} and all messages that are sent out in round $i$ will be processed in round $i+1$.
Additionally, we assume that each node is activated once in each round.
We measure the \emph{congestion} of the system by the maximum number of messages that need to be handled by a node in one round.\footnote{We use $\widetilde{O}(\cdot)$ to hide polylogarithmic factors, i.e., $\widetilde{O}(f(n)) = O(f(n) \cdot polylog(n))$.}
The \emph{injection rate} of a node $v$ is denoted by $\lambda(v)$ and represents the maximum number of heap requests that $v$ is able to generate in each round.
In the asynchronous setting we consider $\lambda(v)$ to be the maximum number of heap requests that $v$ is able to generate between two of its activations.
We assume that $\lambda(v) \in O(poly(n))$ and denote the \emph{maximum injection rate} by $\Lambda = \max_{v \in V}\{\lambda(v)\}$.
We assume that the heap consists of $m$ elements stored by the nodes at any time and the time for which the heap is active is polynomial in $n$.
This implies $m \in O(poly(n))$, because the storage capacities and injection rates of nodes are polynomial in $n$.

\subsection{Basic Notation}
Let $\mathcal E$ be a universe of elements that can be potentially stored by the heap.
Each element $e \in \mathcal E$ is assigned a unique \emph{priority} from a universe $\mathcal P$.
Denote $e$'s priority by $prio(e)$.
We allow different elements to be assigned to the same priority.
Priorities in $\mathcal P$ can be totally ordered via $<$.
Using a tiebreaker to break ties between elements having the same priority, we get a total order on all elements in $\mathcal E$.
A distributed heap supports the following operations:

\begin{itemize}
	\item[-] \ins{e}: Inserts the element $e \in \mathcal E$ into the heap.
	\item[-] \delmin{}: Retrieves the element with minimum priority from the heap or returns $\perp$ if the heap is empty.
	\item[-] \join{}: The node $v$ issuing the operation wants to join the system.
	\item[-] \leave{}: The node $v$ issuing the operation wants to leave the system.
\end{itemize}

As opposed to a standard, sequential heap, it is a non-trivial task to provide protocols for the above operations in our distributed setting: Elements stored in the heap have to be stored by the nodes forming the heap in a \emph{fair} manner, meaning that each node stores $m/n$ elements on expectation, where $n$ denotes the number of nodes forming the heap.
Additionally, in an asynchronous setting, nodes do not have access to local or global clocks and the delivery of messages may be delayed by an arbitrary but finite amount of time.
This may lead to operations violating the heap property known from standard, sequential heaps when considering trivial approaches.
As a consequence, we want to establish a global serialization of the requests ensuring some well-defined semantics without creating bottlenecks in the system, even at a high request rate.
In case the number of priorities $|\mathcal P|$ is constant, we can use \pname{} to guarantee sequential consistency.
If $|\mathcal P|$ is too large, then we can use \pnamep{} which guarantees serializability.

Before we can define these two types of semantics, we need some notation.
Let $\textsc{Ins}_{v,i}$ be the $i^{th}$ \ins{} request that was called by node $v$ and $prio(\textsc{Ins}_{v,i})$ be the priority of the element that is inserted via this request.
Similarly, $\textsc{Del}_{v,i}$ denotes the $i^{th}$ \delmin{} request issued by $v$.
In general, $\textsc{OP}_{v,i}$ denotes the $i^{th}$ (\ins{} or \delmin{}) request issued by $v$.
Let $S$ be the set of all \ins{} and \delmin{} requests issued by all nodes.
We denote a pair $(\textsc{Ins}_{v,i}, \textsc{Del}_{w,j})$ to be \emph{matched}, if $\textsc{Del}_{w,j}$ returns the element that was inserted into the heap via $\textsc{Ins}_{v,i}$.
Let $M$ be the set of all matchings.
Note that not every request has to be matched and thus is not necessarily contained in $M$.
We denote this via $\textsc{OP}_{v,i} \not \in M$.
We are now ready to give formal definitions for sequential consistency and serializability:

\begin{definition} \label{def:semantics}
	A distributed data structure is \emph{serializable} if and only if there is an ordering $\prec$ on the set $S$ so that the distributed execution of all operations of $S$ on the data structure is equivalent to the serial execution w.r.t. $\prec$.
	The data structure is \emph{sequentially consistent} if it is serializable and \emph{locally consistent}, i.e., for all $v \in V$ and $i \in \mathbb{N}$: $\textsc{OP}_{v,i} \prec \textsc{OP}_{v,i+1}$.
\end{definition}

Intuitively, local consistency means that for each single node $v$, the requests issued by $v$ have to come up in $\prec$ in the order they were executed by that node.

In order for our distributed data structure to resemble a heap, we introduce the following additional semantical constrains:

\begin{definition} \label{def:heap_consistency}
	A distributed heap protocol with operations \ins{} and \delmin{} is \emph{heap consistent} if and only if there is an ordering $\prec$ on the set $S$ so that the set of all matchings $M$ established by the protocol satisfies:
	\begin{enumerate}
		\item[(1)] for all $(\textsc{Ins}_{v,i}, \textsc{Del}_{w,j}) \in M: \textsc{Ins}_{v,i} \prec \textsc{Del}_{w,j}$,
		\item[(2)] for all $(\textsc{Ins}_{v,i}, \textsc{Del}_{w,j}) \in M:$ There is no $\textsc{Del}_{u,k} \not \in M$ such that $\textsc{Ins}_{v,i} \prec \textsc{Del}_{u,k} \prec \textsc{Del}_{w,j}$ and
		\item[(3)] for all $(\textsc{Ins}_{v,i}, \textsc{Del}_{w,j}) \in M:$ There is no $\textsc{Ins}_{u,k} \not \in M$ such that $\textsc{Ins}_{u,k} \prec \textsc{Del}_{w,j}$ and $prio(\textsc{Ins}_{u,k}) < prio(\textsc{Ins}_{v,i})$.
	\end{enumerate}
\end{definition}

Intuitively, the three properties have the following meaning: The first property means that elements have to be inserted into the heap before they can be deleted.
The second property means that each \delmin{} request returns a value if there is one in the heap.
The third property means that elements are removed from the heap based on their priority, where elements with minimal priorities are preferred.
Note that this property can be inverted such that our heap behaves like a MaxHeap.

\subsection{Related Work}
Distributed hash tables, initially invented by Plaxton et al.~\cite{DBLP:conf/spaa/PlaxtonRR97} and Karger et al.~\cite{DBLP:conf/stoc/KargerLLPLL97} are the most prominent type of distributed data structure.
Several DHTs for practical applications have been proposed, for example Chord~\cite{DBLP:conf/sigcomm/StoicaMKKB01}, Pastry~\cite{DBLP:conf/middleware/RowstronD01}, Tapestry~\cite{DBLP:journals/jsac/ZhaoHSRJK04} or Cassandra~\cite{DBLP:conf/podc/LakshmanM09}.
We make use of a DHT for storing all elements inserted into the heap and fairly distributing them among all nodes.

We build on the concept from~\cite{DBLP:conf/ipps/FeldmannSS18}, where a sequentially consistent and scalable distributed queue, called \skueue{}, was introduced.
\skueue{} combines the linearized de Bruijn network (LDB)~\cite{DBLP:journals/talg/NaorW07} with a distributed hash table and is able to process batches of enqueue and dequeue operations in $O(\log n)$ rounds w.h.p. via an aggregation tree induced by the LDB.
In \pname{} we show how to extend \skueue{} to construct a distributed heap for a constant amount of priorities, by technically maintaining one distributed queue for each priority.
While \pnamep{} also makes use of the LDB and its induced aggregation tree, it uses a different approach to insert and remove heap elements from the DHT.
This approach involves solving the distributed $k$-selection problem.
Distributed $k$-selection is a classic problem in the sequential setting, but has also been studied in the distributed setting for various types of data structures like cliques~\cite{DBLP:journals/networks/RotemSS86}, rings, meshes or binary trees~\cite{DBLP:conf/podc/Frederickson83}.
Kuhn et. al.~\cite{DBLP:conf/spaa/KuhnLW07} showed a lower bound of $\Omega(D \log_D n)$ on the runtime for any \emph{generic} distributed selection algorithm, where $D$ is the diameter of the network topology.
By 'generic' they mean that the only purpose to access an element is for comparison, which does not hold for our protocol, since we also allow elements to be copied and/or moved to other nodes.
This comes with the advantage that the runtime of our algorithm is only logarithmic in the number of nodes $n$.
Recently Haeupler et. al.~\cite{DBLP:conf/podc/HaeuplerMS18} came up with an algorithm that solves the distributed $k$-selection problem in $O(\log n)$ rounds w.h.p. in the uniform gossip model using $O(\log n)$ bit messages.
This matches our result for distributed $k$-selection in both time and message complexity.
The idea of their algorithm is to compute an approximation for the $k^{\mathit{th}}$ smallest element through sampling and then use this algorithm several times to come up with an exact solution.
While our algorithm for distributed $k$-selection shares some ideas regarding the sampling technique, we are able to find the $k^{\mathit{th}}$ smallest element among $m = poly(n)$ elements distributed over $n$ nodes, whereas the algorithm from~\cite{DBLP:conf/podc/HaeuplerMS18} works only on $n$ elements.

Data structures that are somewhat close to our model are \emph{concurrent} data structures.
In this scenario, multiple nodes issue requests on a (shared) data structure that is stored at a central instance.
The literature ranges from concurrent queues~\cite{DBLP:conf/podc/MichaelS96}, stacks~\cite{DBLP:journals/jpdc/HendlerSY10} to priority queues~\cite{DBLP:conf/spdp/Ayani90, Hunt:1996:EAC:245980.245993, DBLP:journals/jpdc/Johnson94}.
Consider~\cite{DBLP:reference/crc/MoirS04} for a survey on this concept.
All of these data structures are not fully decentralized like ours, as elements of the data structure are stored in shared memory which can be directly modified by nodes.
This gives access to mechanisms such as locks that can prevent multiple nodes to modify the data structure at the same time, which however not only limits scalability but also is vulnerable to \emph{memory contention}, i.e., multiple nodes competing for the same location in memory with only one node being allowed to access the location at any point in time.

Scalable concurrent priority queues have been proposed in~\cite{DBLP:conf/podc/ShavitZ99, DBLP:conf/ipps/ShavitL00}.
In~\cite{DBLP:conf/podc/ShavitZ99} the authors focus on a fixed range of priorities and come up with a technique that is based on \emph{combining trees}~\cite{DBLP:conf/isca/GottliebGKMRS98, Goodman:1989:ESP:68182.68188}, which are similar to the aggregation tree in our work.
However there still is a bottleneck, as the node that is responsible for a combined set of operations has to process them all by itself on the shared memory.
The authors of~\cite{DBLP:conf/ipps/ShavitL00} propose a concurrent priority queue for an arbitrary amount of priorities, where heap elements are sorted in a skiplist.
Their data structure satisfies linearizability but the realization of \delmin{} generates memory contention, as multiple nodes may compete for the same smallest element with only one node being allowed to actually delete it from the heap.

A scalable distributed heap called \emph{SHELL} has been presented by Scheideler and Schmid in~\cite{DBLP:conf/icalp/ScheidelerS09}. 
\emph{SHELL}'s topology resembles the de Bruijn graph and is shown to be very resilient against Sybil attacks.
However, SHELL is about the participants of the system forming a heap and not a distributed data structure that maintains elements.

\subsection{Our Contributions}
In the following we summarize our contributions:

\begin{enumerate}
	\item We propose a distributed protocol for a heap with a constant number of priorities, called \pname{}~(\Cref{sec:heap1}), which guarantees sequential consistency.
	\pname{} is a simple extension of \skueue{}~\cite{DBLP:conf/ipps/FeldmannSS18}, which is a sequentially consistent distributed queue.
	Batches of operations are processed in $O(\log n)$ rounds w.h.p. and congestion of $\widetilde{O}(\Lambda)$.
	\pname{} uses $O(\Lambda \log^2 n)$ bit messages in order to guarantee sequential consistency and to remain scalable.
	Although such a size appears to be quite large, we want to note that it is still superior to a variant in which nodes would have to handle multiple small messages instead of one large message, as this hurts the congestion.
	Also in our protocol each node has to handle only two of these $O(\Lambda \log^2 n)$ bit messages per iteration.
	\item We present a distributed protocol \kselect{}~(\Cref{sec:heap:prio2:k_selection}) that solves the $k$-selection problem and might be of independent interest.
	\kselect{} finds the $k^{\mathit{th}}$ smallest element out of a set of $m \in O(poly(n))$ elements distributed uniformly at random among $n$ nodes in $O(\log n)$ rounds w.h.p., using $O(\log n)$ bit messages and generating only $\widetilde{O}(1)$ congestion on expectation.
	\item We demonstrate how to use \textsc{KSelect} in order to realize \pnamep{}, a distributed heap for arbitrary priorities that guarantees serializability~(\Cref{sec:heap:prio2}).
	 The performance and congestion bounds remain the same as for \pname{}.
	 However, \pnamep{} uses only $O(\log n)$ bit messages independently of the injection rate, which is a huge improvement over \pname{}.
	 Therefore \pnamep{} scales even for a high number of nodes and high injection rates.
	 For applications like job-allocation where local consistency is not that important, it makes sense to use \pnamep{}, but also in scenarios with high injection rates, we recommend using \pnamep{} instead of \pname{} due to the significantly smaller message size in \pnamep{}.
	\item Nodes in \pname{} or \pnamep{} may join or leave the system.
	Through lazy processing, joining or leaving for a node can be done in a constant amount of rounds, whereas the restoration of our network topology is done after $O(\log n)$ rounds w.h.p. for batches of \join{} or \leave{} operations.
	As \join{} and \leave{} work exactly the same as in \skueue{} for both \pname{} and \pnamep{}, we just refer the reader to~\cite{DBLP:conf/ipps/FeldmannSS18} for the details.
\end{enumerate}

Before we describe \pname{} and \pnamep{}, we introduce some preliminaries that serve as the basis for both \pname{} and \pnamep{}.

\section{Preliminaries} \label{sec:preliminaries}

\subsection{Classical de Bruijn Network}
We revise the standard de Bruijn graph~\cite{dB1946} along with its routing protocol. Our network topology that we present afterwards adapts some properties of de Bruijn graphs.

\begin{definition}
	Let $d \in \mathbb{N}$.
	The standard ($d$-dimensional) \emph{de Bruijn graph} consists of nodes $(x_1, \ldots ,x_d) \in \{0,1\}^d$ and edges $(x_1, \ldots ,x_d)$ $\rightarrow$ $(j,x_1, \ldots ,x_{d-1})$ for all $j \in \{0,1\}$.
\end{definition}

One can route a packet via bitshifting from a source $s \in \{0,1\}^d$ to a target $t \in \{0,1\}^d$ by adjusting exactly $d$ bits.
For example for $d = 3$, we route from $s = (s_1,s_2,s_3)$ to $t = (t_1,t_2,t_3)$ via the path $((s_1,s_2,s_3), (t_3,s_1,s_2),(t_2,t_3,s_1),(t_1,t_2,t_3))$.

\subsection{Aggregation Tree}
We use the \emph{aggregation tree} that is defined in the following to interconnect the nodes.
The aggregation tree is also able to emulate routing in the standard de Bruijn graph.

\begin{lemma} \label{lemma:aggregation_tree}
	There exists a tree $T = (V,E)$ (which we call \emph{aggregation tree} in this paper) that satisfies the following properties:
	\begin{itemize}
		\item[(i)] Each inner node $v$ has a parent node $p(v)$ and up to two child nodes (denoted by the set $C(v)$). $T$ has height $O(\log n)$ w.h.p.
		\item[(ii)] $T$ embeds a distributed hash table (DHT) supporting operations \textsc{Put}$(k$, $e)$ to store the element $e$ at the node $v$ that maintains the key $k$ and \textsc{Get}$(k$, $v)$ to retrieve the element that is stored under the key $k$ and deliver it to node $v$.
		\item[(iii)] The DHT supports \textsc{Put} and \textsc{Get} requests in $O(\log n)$ rounds w.h.p.
		\item[(iv)] Elements in the DHT are stored uniformly among all nodes, i.e., if the DHT contains $m$ elements, then each node has to store $m/n$ elements on expectation.
		\item[(v)] Any routing schedule with dilation $D$ and congestion $C$ in the $d$-dimensional de Bruijn graph can be emulated by $T$ with $n \geq 2$ nodes with dilation $O(D + \log n)$ and congestion $\widetilde{O}(C)$ w.h.p.	
	\end{itemize}
\end{lemma}

Denote the root node of the tree as the \emph{anchor}.
The aggregation tree can be used to aggregate certain values to the anchor.
We call this process an \emph{aggregation phase}.
Values can be \emph{combined} with other values at each node.
For instance, to determine the number of nodes that participate in the tree, each node initially holds the value $1$.
We start at the leaf nodes, which send their value to their parent nodes upon activation.
Once an inner node $v$ has received all values $k_1,\ldots,k_c \in \mathbb{N}$ from its $c$ child nodes, upon activation it combines these by adding them to its own value, i.e., by computing $1 + \sum_{i=1}^c k_i$.
Afterwards $v$ sends the result to its parent node.
Once the anchor has combined the values of its child nodes with its own value it knows $n$.
We make heavy use of aggregation phases in both \pname{} and \pnamep{}.
Due to \Cref{lemma:aggregation_tree}(i) it is easy to see that an aggregation phase finishes after $O(\log n)$ rounds w.h.p.

We assume for the rest of the paper that the nodes are arranged in such an aggregation tree and explain in \Cref{appendix:aggregation_tree} how it can be built.

\section{Constant Priorities} \label{sec:heap1}
In this section we introduce the \pname{} protocol for a distributed heap with a constant number of priorities, i.e., $\mathcal{P} = \{1,\ldots,c\}$ for a constant $c \in \mathbb{N}$.
\pname{} is an extension of \skueue{} and is able to achieve the same runtimes for operations as well as the same semantic, i.e., sequential consistency.
Throughout this paper, a \emph{heap operation} is either an \ins{} or a \delmin{} request.
The main challenge to guarantee sequential consistency lies in the fact that messages may outrun each other, since we allow fully asynchronous computations and non-FIFO message delivery.
Another problem we have to solve is that the rate at which nodes issue heap requests may be very high.
As long as we process each single request one by one, scalability cannot be guaranteed.

The general idea behind \pname{} is the following: First, we use the aggregation tree to aggregate batches of heap operations to anchor (Phase 1).
The anchor then assigns a unique position for each heap operation such that sequential consistency is fulfilled (Phase 2) and spreads all positions for the heap operations over the aggregation tree afterwards (Phase 3).
Heap elements are then inserted into or fetched from the DHT according to the given position (Phase 4).
We describe this approach in greater detail now.

\subsection{Operation Batch}
Whenever a node initiates a heap operation, it buffers it in its local storage.
We are going to represent the sequence of buffered heap operations by a \emph{batch}:

\begin{definition}[Batch] \label{def:batch}
	A \emph{batch} $B$ (of length $k$) is a sequence $(i_1,d_1,$ $\ldots,$ $i_k,d_k)$, for which it holds that for all $j \in \{1,\ldots,k\}$, $i_j = (i_{j,1},\ldots,i_{j,|\mathcal P|}) \in \mathbb{N}^{|\mathcal P|}$ consists of values for each $p \in \mathcal P$ representing the number of elements with priority $p$ to be inserted.
	For all $j \in \{1,\ldots,k\}$, $d_j \in \mathbb{N}$ represents the number of \delmin{} operations.
\end{definition} 

We \emph{combine} two batches $B_1=(i_1,d_1,\ldots,i_k,d_k)$ and $B_2=(i_1',d_1',\ldots,i_k',d_k')$ by computing $B = (i_1'',d_1'',\ldots,i_k'',d_k'')$, where $i_j''=(i_{j,1}+i_{j,1}',\ldots,i_{j,k}+i_{j,k}')$ and $d_j''=d_j+d_j'$.
Note that in case $B_1$ and $B_2$ are of different length, we just fill up the smaller batch with zeros.
If a batch $B$ is the combination of batches $B_1,\ldots,B_k$, then we denote $B_1,\ldots,B_k$ as \emph{sub-batches}.

\subsection{Protocol SKEAP}
We are now ready to describe our approach for processing heap operations in detail, dividing it into $4$ phases.
\Cref{algo:skeap} summarizes \pname.

\begin{algorithm}[ht]
\caption{Protocol \pname}
\label{algo:skeap}
\begin{algorithmic}[1]
	\Statex \textbf{Phase 1} (Executed at each node $v$)
	\State \pushcode[0] Create batch $v.B$
	\State \pushcode[0] Wait until each $w \in C(v)$ has sent its batch to $v$
	\State \pushcode[0] Combine batches $w.B$ of all $w \in C(v)$ with $v.B$
	\State \pushcode[0] Send combined batch $v.B^+ = (i_1,d_1,\ldots,i_k,d_k)$ to $p(v)$
	\Statex
	\Statex \textbf{Phase 2} (Only local computation at the anchor)
	\State \pushcode[0] Let $(i_1,d_1,\ldots,i_k,d_k)$ be the combined batch from Phase 1
	\State \pushcode[0] \textbf{for} $j \in \{1,\ldots,k\}$ \textbf{do}
	\State \pushcode[0] \ \ \ \ Compute position intervals for $i_j$
	\State \pushcode[0] \ \ \ \ Compute position intervals for $d_j$
	\Statex
	\Statex \textbf{Phase 3} (Executed at each node $v$)
	\State \pushcode[0] Wait until $p(v)$ has sent position intervals $I$ to $v$
	\State \pushcode[0] Decompose $I$ into $I_v$ and $I_w$ for each $w \in C(v)$ 
	\State \pushcode[0] \textbf{for} all $w \in C(v)$ \textbf{do}
	\State \pushcode[0] \ \ \ \ Send $I_w$ to $w$
	\Statex
	\Statex \textbf{Phase 4} (Executed at each node $v$)
	\State \pushcode[0] Generate DHT operations based on $I_v$
\end{algorithmic}
\end{algorithm}

Whenever a node $v$ generates a new \ins{} or \delmin{} request it stores the request in a local queue that acts as a buffer.
At the beginning of the first phase of \pname{} each node $v$ generates a snapshot of the contents of its queue and represents it as a batch $v.B$.
For example, a snapshot consisting of operations $\textsc{Insert}(e_1),\textsc{Insert}(e_2), \textsc{DeleteMin}(),\textsc{Insert}(e_3)$ and $\textsc{DeleteMin}()$ (in that specific order) with $prio(e_1) = 1$, $prio(e_2)=1$ and $prio(e_3) = 2$ is represented by the batch $((2,0),1,(0,1),1)$.
By doing so, the batch $v.B$ respects the local order in which heap operations are generated by $v$, which is important for guaranteeing sequential consistency.

\subsubsection{Phase 1: Aggregating Batches}
In the first phase we aggregate batches of all nodes up to the anchor via an aggregation phase.
For this each node $v$ waits until it has received all batches from its child nodes $w \in C(v)$ and then combines them together with its own batch $v.B$ upon activation.
The resulting batch, denote it by $v.B^+$, is then sent to the parent $p(v)$ of $v$ in the aggregation tree.
Additionally, $v$ memorizes the sub-batches it received from its child nodes, as it needs them to perform the correct interval decomposition in Phase~3.
At the end of the first phase the anchor $v_0$ possesses a batch $v_0.B^+$ that is the combination of all sub-batches $v.B$.

\subsubsection{Phase 2: Assigning Positions}
We only perform local computations at the anchor $v_0$ in this phase.
The anchor maintains variables $v_0.first_p \in \mathbb{N}_0$ and $v_0.last_p \in \mathbb{N}_0$ for each priority $p \in \mathcal P$, such that the invariant $v_0.first_p \leq v_0.last_p + 1$ holds at any time for any $p \in \mathcal P$.
The interval $[v_0.first_p, v_0.last_p]$ represents the positions that are currently occupied by elements with priority $p$.
This implies that the heap contains $v_0.last_p - v_0.first_p + 1$ elements of priority $p$.
Denote this number as the \emph{cardinality} of the interval $[v_0.first_p, v_0.last_p]$, i.e., $|[v_0.first_p, v_0.last_p]| = v_0.last - v_0.first + 1$ and denote an interval with cardinality $0$ as \emph{empty}.

We first describe the actions of the anchor for one (fixed) priority $p$ and thus assume that for the batch $v_0.B^+ = (i_1,d_1,\ldots,i_k,d_k)$ it holds $i_j \in \mathbb{N}_0$ instead of $i_j$ being a (sub-)vector.
Based on its variables $v_0.first_p, v_0.last_p$, $v_0$ computes a collection of intervals $I_1,D_1,\ldots,I_k,D_k$ as follows: For the $j^{\mathit{th}}$ insert entry $i_j \in \mathbb{N}_0$ of $v_0.B^+$, $v_0$ sets $I_j$ to $[v_0.last_p + 1, v_0.last_p + i_j]$ and increases $v_0.last_p$ by $i_j$ afterwards.
For the $j^{\mathit{th}}$ delete entry $d_j \in \mathbb{N}_0$ of $v_0.B^+$, $v_0$ sets $D_j$ to $[v_0.first_p, \min\{v_0.first_p + d_j - 1, v_0.last_p\}]$ and updates $v_0.first_p$ to $\min\{v_0.first_p + d_j, v_0.last_p + 1\}$ afterwards.

To extend this approach to $|\mathcal P|$ priorities we just let the anchor perform the above computations on the corresponding interval $[v_0.first_p,v_0.last_p]$ for each priority $p \in \mathcal P$.
Also note that if $d_j$ elements should be removed from the heap, then the anchor first considers the most prioritized non-empty interval $[v_0.first_p, v_0.last_p]$ and computes the interval $D_j$ as described above.
If the cardinality of $D_j$ should be less than $d_j$, then the anchor looks for the next non-empty interval in the total order of priorities to compute another interval of positions to add to $D_j$.
The anchor stops if all intervals $[v_0.first_p, v_0.last_p]$ are empty or the cardinality of $D_j$ is equal to $d_j$.
Observe that this approach technically leads to each $I_j$ and $D_j$ being a collection of at most $|\mathcal P|$ intervals.
By doing so, the anchor assigned a vector of intervals to each entry in $v_0.B^+$, implying that we can assign a pair $(p,pos) \in \mathcal P \times \mathbb{N}$ to each single heap operation represented by the batch (which is part of the next phase).

\subsubsection{Phase 3: Decomposing Position Intervals}
Once $v_0$ has computed all required intervals $I_1,\ldots,I_k$ for a batch, it starts the third phase, in which these intervals are decomposed and broadcasted over all nodes in the aggregation tree.
When a node $v$ in the tree receives $I_1,\ldots,I_k$, it decomposes the intervals with respect to $v.B$ and the sub-batches of its child nodes it received in Phase~1 (recall that $v$ has memorized these).
Once each sub-batch has been assigned to a collection of (sub\mbox{-})intervals, we send out these intervals to the respective child nodes in $C(v)$.
Applying this procedure in a recursive manner down the aggregation tree yields an assignment of a pair $(p,pos) \in \mathcal P \times \mathbb{N}$ to all \ins{} and \delmin requests.
\Cref{fig:skeap_example} illustrates the first $3$ phases of \pname{}.

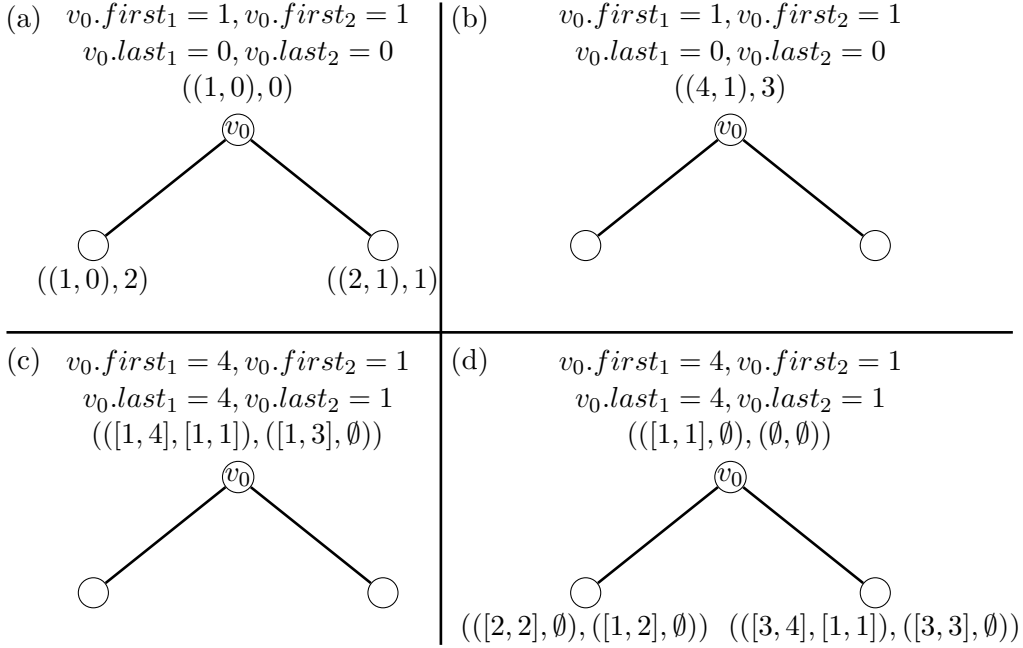
\begin{figure*}[ht]
 	\centering
 	\begin{tikzpicture}[scale=0.77, main node/.style={circle,draw,align=center, minimum size=0.4cm, inner sep=0pt}]	
     	\node[main node, label={[yshift=-1cm]$((1,0),2)$}] (A1) at (-0.5,0) {};
     	\node[main node, label={[yshift=0cm]$((1,0),0)$}, label={[yshift=0.5cm]$v_0.last_1 = 0, v_0.last_2 = 0$}, label={[yshift=1cm]$v_0.first_1 = 1, v_0.first_2 = 1$}] (B1) at (2,2) {$v_0$};
     	\node[main node, label={[yshift=-1cm]$((2,1),1)$}] (C1) at (4.5,0) {};
     	
     	\draw[-, line width=1pt] (A1) to (B1);
     	\draw[-, line width=1pt] (C1) to (B1);

     	\node[main node] (A2) at (8,0) {};
     	\node[main node, label={[yshift=0cm]$((4,1),3)$}, label={[yshift=0.5cm]$v_0.last_1 = 0, v_0.last_2 = 0$}, label={[yshift=1cm]$v_0.first_1 = 1, v_0.first_2 = 1$}] (B2) at (10.5,2) {$v_0$};
     	\node[main node] (C2) at (13,0) {};
     	
     	\draw[-, line width=1pt] (A2) to (B2);
     	\draw[-, line width=1pt] (C2) to (B2);

     	\node[main node] (A3) at (-0.5,-6) {};
     	\node[main node, label={[yshift=0cm]$(([1,4],[1,1]),([1,3],\emptyset))$}, label={[yshift=0.5cm]$v_0.last_1 = 4, v_0.last_2 = 1$}, label={[yshift=1cm]$v_0.first_1 = 4, v_0.first_2 = 1$}] (B3) at (2,-4) {$v_0$};
     	\node[main node] (C3) at (4.5,-6) {};
     	
     	\draw[-, line width=1pt] (A3) to (B3);
     	\draw[-, line width=1pt] (C3) to (B3);

     	\node[main node, label={[yshift=-1cm]$(([2,2],\emptyset),([1,2],\emptyset))$}] (A4) at (8,-6) {};
     	\node[main node, label={[yshift=0cm]$(([1,1],\emptyset),(\emptyset,\emptyset))$}, label={[yshift=0.5cm]$v_0.last_1 = 4, v_0.last_2 = 1$}, label={[yshift=1cm]$v_0.first_1 = 4, v_0.first_2 = 1$}] (B4) at (10.5,-4) {$v_0$};
     	\node[main node, label={[yshift=-1cm]$(([3,4],[1,1]),([3,3],\emptyset))$}] (C4) at (13,-6) {};
     	
     	\draw[-, line width=1pt] (A4) to (B4);
     	\draw[-, line width=1pt] (C4) to (B4);
     	\draw[-, line width=1pt] (-2,-1.5) to (15.375,-1.5);
     	\draw[-, line width=1pt] (5.5,4.2) to (5.5,-6.9);
     	
     	\node[] (A3) at (-1.7,3.9) {(a)};
     	\node[] (A3) at (6,3.9) {(b)};
     	\node[] (A3) at (-1.7,-2) {(c)};
     	\node[] (A3) at (6,-2) {(d)};
 	\end{tikzpicture}
	\caption{Illustration of the phases for \pname{} for $n = 3$ nodes, exemplary batches for each node and $\mathcal P=\{1,2\}$. (a) Before Phase 1. (b) After Phase 1 and before Phase 2. (c) After Phase 2 and before Phase 3. (d) After Phase 3.}
 	\label{fig:skeap_example}
\end{figure*}

\subsubsection{Phase 4: Updating the DHT}
Now that a node $v$ knows the unique pair $(p,pos) \in \mathcal P \times \mathbb{N}$ for each of its heap operations, it starts generating operations on the DHT.
For a request \ins{e} that got assigned to pair $(p,pos)$, $v$ generates a request \textsc{Put}($h(p,pos)$, $e$) to insert $e$ into the DHT.
Here $h: \mathcal P \times \mathbb{N} \rightarrow \mathbb{N}$ is a publicly known pseudorandom hash function that generates the key under which $e$ should be stored.
This finishes the \ins{e} request.

For a \delmin request that got assigned to pair $(p,pos)$, $v$ generates a request \textsc{Get}($h(p,pos)$, $v$).
Note that for such a \delmin request there exists a corresponding \textsc{Put}($h(p,pos)$, $e$) request.
As $h$ is pseudorandom, both of these requests are guaranteed to meet at the same node.
Due to asynchronicity, it may happen that a \textsc{Get} request arrives at the correct node in the DHT \emph{before} the corresponding \textsc{Put} request.
In this case the \textsc{Get} request waits at that node until the corresponding \textsc{Put} request has arrived, which will eventually happen.

Once a node has generated all its DHT requests, it switches again to Phase~1, in order to process the next batch of heap operations.

\subsection{Results}
The following theorem summarizes our results for \pname{}.

\begin{theorem} \label{theorem:skeap}
	\pname{} implements a distributed heap with the following properties:
	\begin{enumerate}
		\item \pname{} is fair.
		\item \pname{} satisfies sequential consistency and heap consistency.
		\item The number of rounds, needed to process heap requests successfully is $O(\log n)$ w.h.p.
		\item The congestion of \pname{} is at most $\widetilde{O}(\Lambda)$.
		\item Messages in \pname{} have size of at most $O(\Lambda \log^2 n)$ bits.
	\end{enumerate}
\end{theorem}

We first show that \pname{} guarantees sequential consistency.
In order to do this we define a total order on all \ins{} and \delmin{} requests and show that the chosen order satisfies the properties given in \Cref{def:semantics}.

In order to define $\prec$ we want to assign a unique value to each heap request $\mathit{OP}$, indicated by $value(\mathit{OP}) \in \mathbb{N}$.
Informally, $value(\mathit{OP})$ is the number of requests that the anchor has ever processed in Phase~2 up to (and including) $\mathit{OP}$.
More formally, recall \Cref{def:batch} and assume w.l.o.g. that $i_j = \sum_{p = 1}^{|\mathcal P|} i_{j,p}$.
Every time a node issues a heap operation $\mathit{OP}$, it increments a value $op_j$ ($=i_j$ or $d_j$) of some batch $B$ that would be created at the start of Phase~1, i.e., $op_j' = op_j + 1$.
Initially, set $value(\mathit{OP}) = \sum_{k=1}^{j-1} op_k + op_j'$.
Whenever we combine two batches $B_1 = (i_1,d_1,\ldots,i_k,d_k)$ and $B_2 = (i_1',d_1',\ldots,i_k',d_k')$ on their way to the anchor as part of Phase~1, we do the following: We choose one of the batches to be the first~(let this be $B_1$ here) and the other one to be the second ($B_2$) and modify all values for operations $\mathit{OP}$ that are part of the second batch.
If $value(\mathit{OP})$ is the value for $\mathit{OP}$ before combining the batches, we set the new value for $\mathit{OP}$ to $value'(\mathit{OP}) = (\sum_{j=1}^k i_j+d_j) + value(\mathit{OP})$.
We let the anchor maintain a (virtual) variable $count \in \mathbb{N}$, which is set to $1$ initially.
Every time the anchor processes a batch $B^+ = (i_1,d_1,\ldots,i_k,d_k)$ in Phase~2, we set the value for each request $\mathit{OP}$ to $value(\mathit{OP}) = count + value(\mathit{OP})$.
Afterwards we increase $count$ in preparation for the next incoming batch by setting $count = count + \sum_{j=1}^k i_j+d_j$.
Notice that this approach leads to each operation getting a unique value and thus, we can define the total order $\prec$ by sorting all operations according to their values.
By definition of our protocol this order resembles the exact order in which the anchor processes the batch that corresponds to the operations, implying that \pname{} is serializable.

The following lemmas follows directly from the protocol description:

\begin{lemma} \label{lemma:const:1}
	Let $(\textsc{Ins}_{v,i}, \textsc{Del}_{w,j}) \in M$.
	Then both operations got assigned to the same pair $(p, pos)$ from the anchor.
\end{lemma}

\begin{lemma}\label{lemma:const:skueue}
	If two operations $\textsc{Ins}_{v,i}$ and $\textsc{Del}_{w,j}$ got assigned to the same pair from the anchor, then $\textsc{Ins}_{v,i} \prec \textsc{Del}_{w,j}$.
\end{lemma}

We are now ready to prove our main result of this section:

\begin{lemma}
	\pname{} implements a fair distributed heap that is sequentially consistent and heap consistent.
\end{lemma}

\begin{proof}
	Fairness is clear, since we make use of a pseudorandom hash function in order to distribute heap elements uniformly over all nodes.
	
	Next we show heap consistency for \pname{}.	
	We consider each property of \Cref{def:semantics} individually and show that the order defined above by $\prec$ fulfills that requirement.
	First, it follows from \Cref{lemma:const:1} that if $(\textsc{Ins}_{v,i}, \textsc{Del}_{w,j}) \in M$, then both operations got assigned the same pair $(p, pos)$ by the anchor, i.e., both requests operate on the same interval $[v_0.first_p,v_0.last_p]$.
	Thus, it follows from \Cref{lemma:const:skueue} that $\textsc{Ins}_{v,i} \prec \textsc{Del}_{w,j}$.
	
	For the second property let $(\textsc{Ins}_{v,i}, \textsc{Del}_{w,j}) \in M$ and assume there is $\textsc{Del}_{u,k} \not \in M$ such that $\textsc{Ins}_{v,i} \prec \textsc{Del}_{u,k} \prec \textsc{Del}_{w,j}$.
	From $\textsc{Del}_{u,k} \not \in M$ it follows that $\textsc{Del}_{u,k}$ returns $\perp$.
	Since $\textsc{Ins}_{v,i} \prec \textsc{Del}_{u,k}$, there have to be other \delmin{} operations that lie between $\textsc{Ins}_{v,i}$ and $\textsc{Del}_{u,k}$ in $\prec$ with one of them being matched to $\textsc{Ins}_{v,i}$.
	But this contradicts the fact that $\textsc{Ins}_{v,i}$ got matched to $\textsc{Del}_{w,j}$ which is processed after $\textsc{Del}_{u,k}$ by the anchor.

	For the third property let $(\textsc{Ins}_{v,i}, \textsc{Del}_{w,j}) \in M$ and assume there is a $\textsc{Ins}_{u,k} \not \in M$ such that $\textsc{Ins}_{u,k} \prec \textsc{Del}_{w,j}$ and $prio(\textsc{Ins}_{u,k}) < prio(\textsc{Ins}_{v,i})$.
	Because of $\textsc{Ins}_{u,k} \prec \textsc{Del}_{w,j}$ it follows that at the time the anchor processes $\textsc{Del}_{w,j}$, the heap contains both the elements that have been inserted by $\textsc{Ins}_{u,k}$ and $\textsc{Ins}_{v,i}$.
	But then, by definition of Phase~2 of \pname{}, the anchor should have assigned $\textsc{Ins}_{u,k}$ to $\textsc{Del}_{w,j}$, as it has to be prioritized, which contradicts our assumption of $\textsc{Ins}_{v,i}$ being matched to $\textsc{Del}_{w,j} \in M$.
	
	By the argumentation above we know that \pname{} is at least serializable.
	Additionally, local consistency is directly satisfied by the way we defined $\prec$ and thus \pname{} is also sequentially consistent.
\end{proof}

As \pname{} consists of the same phases as \skueue{}, the following corollary stating the runtime bound for heap operations follows from~\cite{DBLP:conf/ipps/FeldmannSS18}:

\begin{corollary}\label{cor:skeap:runtime}
	Assume a node $v \in V$ has stored an arbitrary amount of heap requests in its local queue.
	The number of rounds, needed to process all requests successfully is $O(\log n)$ w.h.p.
\end{corollary}

Now we prove the bound for the congestion of \pname{}:

\begin{lemma} \label{theorem:skeap:congestion}
	The congestion of \pname{} is at most $\widetilde{O}(\Lambda)$.
\end{lemma}

\begin{proof}
	At the beginning Phase~1, each node $v$ has at most $\lambda(v) \cdot O(\log n) = \widetilde{O}(\lambda(v))$ heap requests buffered, since the previous execution lasted for $O(\log n)$ rounds (\Cref{cor:skeap:runtime}) and $v$ could have generated at most $\lambda(v)$ requests per round.
	For each of those requests, $v$ generates a single DHT operation (either \textsc{Put} or \textsc{Get}), resulting in $v$ having to process $\widetilde{O}(\lambda(v))$ requests at once.
	Since each DHT operation needs $O(\log n)$ rounds w.h.p. (\Cref{lemma:aggregation_tree}(iii)) to finish and the aggregation tree only generates congestion up to a polylogarithmic factor (\Cref{lemma:aggregation_tree}(v)), the lemma follows.
\end{proof}

Finally we show the bound on the size of messages for \pname{}:

\begin{lemma}
	Messages in \pname{} have size of at most $O(\Lambda \log^2 n)$ bits.
\end{lemma}

\begin{proof}
	In one round each node $v$ may generate up to $\Lambda$ new heap requests.
	If these request alternate between \ins{} and \delmin{} operations, then the corresponding batch has size $O(\Lambda)$ bits.
	Due to \Cref{cor:skeap:runtime} each node may repeat the above procedure for $\log n$ rounds until the whole batch is processed.
	Thus when combining all batches of all nodes, the resulting batch has size $O(\Lambda \log^2 n)$ bits (the batch contains $O(\Lambda \log n)$ entries, each entry is a number in $O(n)$, so it has to be encoded via $O(\log n)$ bits).
\end{proof}

\section{Distributed k-Selection} \label{sec:heap:prio2:k_selection}
In order to construct the protocol \pnamep{} for a heap with arbitrary priorities, we need to solve the distributed $k$-selection problem.
The protocol \textsc{KSelect} that we propose in this section might be of independent interest.
Throughout this section we are given an aggregation tree of $n$ nodes with $m \in \mathbb{N}$ elements distributed uniformly among all nodes, i.e., each node $v$ stores $m/n$ elements on expectation.
Denote by $v.E$ the set of elements stored at node $v$.
Recall that the storage capacity of each node is polynomial in $n$, so $m \in O(poly(n))$, i.e., $m \leq n^q$ for a constant $q \in \mathbb{N}$.
Consider the ordering $e_1 < \ldots < e_m$ of all elements stored in the heap according to their priorities.
We denote the \emph{rank} of an element $e_i$ in this ordering by $rank(e_i) = i$.
As we will use \textsc{KSelect} for \pnamep{}, we assume that the set of priorities is larger than constant now, i.e., $\mathcal P = \{1,\ldots,n^q\}$.

\begin{definition}
	Given a value $k \in \mathbb{N}$, the \emph{distributed $k$-selection problem} is the problem of determining the $k^{\mathit{th}}$ smallest element out of a set of $m=O(poly(n))$ elements, i.e., the element $e$ with $rank(e) = k$.
\end{definition}

For scalability reasons we restrict the size of a message to at most $O(\log n)$ bits.

Each node $v$ maintains a set $v.C \subseteq v.E$ that represents the remaining \emph{candidates} for the $k^{\mathit{th}}$ smallest element at $v$.
Denote the overall set of candidates by $C = \bigcup_{v \in V} v.C$ and the number of remaining candidates by $N = |C|$.
Initially each node $v$ sets $v.C$ to $v.E$, which leads to $N = m$.
We assume that the anchor initially knows of values $n$ and $m$ (and thus also knows an appropriate value for $q$) as these can easily be computed via a single aggregation phase.
The anchor keeps track of values $N$ and $k$ throughout all phases of our protocol via variables $v_0.N$ and $v_0.k$.
Note that once we are able to reduce $N$, we also have to update the value for $k$, because removing a single candidate with rank less than $k$ implies that we only have to search for the $k-1^{\mathit{th}}$ smallest element for the remaining candidates.
Throughout the analysis we use the notation $\exp(x)$ instead of $e^x$.

We dedicate this section to the proof of the following theorem:

\begin{theorem} \label{theorem:k_selection:runtime}
	\kselect{} solves the distributed $k$-selection problem in $O(\log n)$ rounds and congestion $\widetilde{O}(1)$ w.h.p. using $O(\log n)$ bit messages.
\end{theorem}

\kselect{} works in three phases: In the first phase we perform a series of $\log(q) + 1$ iterations on the aggregation tree in order to reduce the number of possible candidates from $n^q$ to $O(n^{3/2} \cdot \log n)$ elements.
The second phase further reduces this number to $O(\sqrt{n})$ candidates via aggregating samples in parallel for $\sqrt{n}$ elements.
In the last phase we directly compute the $k^{\mathit{th}}$ smallest element out of the remaining $O(\sqrt{n})$ candidates.
\Cref{algo:k_select} sums up our protocol.

\begin{algorithm}[ht]
\caption{Protocol \kselect{} for distributed $k$-Selection}
\label{algo:k_select}
\begin{algorithmic}[1]
	\Statex \textbf{Input}: $n$, $m = n^q$, $k$
	\Statex \textbf{Output}: $e_k \in \mathcal E$ with $rank(e_k)=k$
	\Statex 
	\Statex \textbf{Initialization}
	\State \pushcode[0] $v_0.N \gets m$
	\State \pushcode[0] $v_0.k \gets k$	
	\Statex
	\Statex \textbf{Phase 1} (Repeat $\log(q) + 1$ times)
	\State \pushcode[0] Propagate $n$, $v_0.k$ to all nodes
	\State \pushcode[0] Compute $v.P_{\mathit{min}},v.P_{\mathit{max}} \in \mathcal P$ at each node $v \in V$
	\State \pushcode[0] Compute $P_{\mathit{min}} = \min_{v \in V}\{v.P_{\mathit{min}}\}$ and
	\Statex\pushcode[0] $P_{\mathit{max}} = \max_{v \in V}\{v.P_{\mathit{max}}\}$
	\State \pushcode[0] Remove candidates with priorities not in $[P_{\mathit{min}},P_{\mathit{max}}]$
	\State \pushcode[0] Update $v_0.k$, $v_0.N$
	\Statex
	\Statex \textbf{Phase 2} (Repeat until $v_0.N \leq \sqrt{n}$)
	\State \pushcode[0] Propagate $n$, $v_0.N$ to all nodes
	\State \pushcode[0] For each $e \in C$: Include $e$ into $C'$ with probability $\sqrt{n}/N$
	\State \pushcode[0] Sort candidates $c_1,\ldots,c_{n'} \in C'$ based on their priority
	\State \pushcode[0] Fix $\delta \in \Theta(\sqrt{\log n} \cdot \sqrt[4]{n})$
	\State \pushcode[0] Determine $c_l,c_r \in C'$ with $l = \lfloor k \frac{n'}{N} - \delta\rfloor$ and $r = \lceil k \frac{n'}{N} + \delta\rceil$
	\State \pushcode[0] Remove candidates with priorities not in $[prio(c_l),prio(c_r)]$
	\State \pushcode[0] Update $v_0.k$, $v_0.N$
	\Statex
	\Statex \textbf{Phase 3}
	\State \pushcode[0] Sort remaining candidates based on their priority
	\State \pushcode[0] \Return $e_k$
\end{algorithmic}
\end{algorithm}

\subsection{Phase 1: Sampling}
The first phase involves $\log(q) + 1$ iterations: At the start of each iteration, the anchor propagates the values of $k$ and $n$ to all nodes via an aggregation phase.
Then each node $v$ computes the priorities of the $\lfloor k/n \rfloor^{\mathit{th}}$ and the $\lceil k/n \rceil^{\mathit{th}}$ smallest candidates of $v.C$.
Let these priorities be denoted by $v.P_{\mathit{min}}$ and $v.P_{\mathit{max}}$.
Nodes then aggregate these values up to the anchor, such that in the end the anchor gets values $P_{\mathit{min}} = \min_{v \in V} \{v.P_{\mathit{min}}\}$ and $P_{\mathit{max}} = \max_{v \in V} \{v.P_{\mathit{max}}\}$.
The anchor then instructs all nodes $v$ in the aggregation tree to remove all candidates from $v.C$ with priority less than $P_{\mathit{min}}$ or larger than $P_{\mathit{max}}$ and aggregate the overall number $k'$ of candidates less than $P_{\mathit{min}}$ and the number $k''$ of candidates larger than $P_{\mathit{max}}$ up to the anchor, such that it can update $v_0.N$~(by setting $v_0.N$ to $v_0.N - (k'+k'')$) and $v_0.k$~(by setting $v_0.k$ to $v_0.k - k'$).
We obtain the following two lemmas:

\begin{lemma}[Correctness]
	Let $e_k \in C$ be the $k^{\mathit{th}}$ smallest element, i.e., $rank(e_k) = k$.
	Then $P_{\mathit{min}} \leq prio(e_k) \leq P_{\mathit{max}}$.
\end{lemma}

\begin{proof}
	We only show $P_{\mathit{min}} \leq prio(e_k)$ as the proof for $prio(e_k) \leq P_{\mathit{max}}$ works analogously.
	Assume to the contrary that $P_{\mathit{min}}> prio(e_k)$.
	Then each node $v \in V$ has chosen $v.P_{\mathit{min}}$ with $v.P_{\mathit{min}} > prio(e_k)$.
	This implies that the number of elements with rank less or equal than $k$ is at most $(\lfloor k/n \rfloor - 1) \cdot n \leq (k/n - 1) \cdot n \leq k-1$, which is a contradiction.
\end{proof}

\begin{lemma} \label{lemma:k_select:first_iteration}
	After $\log(q) + 1$ iterations of the first phase, $N \in O(n^{3/2} \cdot \log n)$ w.h.p.
\end{lemma}

\begin{proof}
	First we want to compute how many candidates are left in variables $v.C$ after a single iteration of our protocol: Let $X_i$ be the event that the candidate $c_i$ with $rank(c_i) = i$ is stored at node $v$ for a fixed $v \in V$.
	Then $\Pr[X_i = 1] = 1/n$.
	Let $X = \sum_{i=1}^k X_i$.
	Then $E[X] = k/n$.
	$X$ denotes the number of candidates stored at $v$ with rank within $[1,k]$.
	We show that the rank of the $\lfloor k/n \rfloor^{\mathit{th}}$ smallest candidate in $v.C$ deviates from $k$ by only $O(\sqrt{nk\log n})$ w.h.p.
	This holds, because when using Chernoff bounds we get that \[\Pr\left[X \leq (1-\varepsilon) \frac{k}{n} \right] \leq exp\left(-\varepsilon^2 \frac{k}{2n}\right)\leq n^{-c}\] for $\varepsilon = \sqrt{(c \log n) \cdot 2n/k}$ and a constant $c > 0$.
	So with high probability, each node $v$ has at least $(1-\varepsilon)\cdot \frac{k}{n}$ candidates with rank within $[1,k]$ stored in $v.C$.
	It follows that the rank of the $\lfloor k/n \rfloor^{\mathit{th}}$ smallest candidate chosen by $v$ is at least $(1-\varepsilon)\cdot k$ w.h.p.
	By the union bound we know that w.h.p. the rank of the candidate with priority $P_{\mathit{min}}$ is at least $(1-\varepsilon)\cdot k$, so it deviates from $k$ by at most \[k \cdot \varepsilon = k \cdot \sqrt{(c \log n) \cdot 2n/k} = O(\sqrt{nk\log n}).\]
	The same argumentation can be used to show that the rank of the candidate with priority $P_{\mathit{max}}$ deviates from $k$ by no more than $O(\sqrt{nk\log n})$ w.h.p.
	
	So the number of candidates that are left after a single iteration of the first phase is at most $O(\sqrt{nk\log n})$ w.h.p.
	Replacing $k$ with its maximum value $n^q$ yields an upper bound of $O(n^{(q+1)/2} \cdot \sqrt{\log n})$ remaining candidates.
	Thus, after  $\log(q) + 1$ iterations of the first phase, the overall number of candidates left is equal to $O(n^{3/2} \prod_{i=1}^{\log(q) + 1} \sqrt[2^i]{\log n}) = O(n^{3/2} \cdot \log n)$.
\end{proof}

\subsection{Phase 2a: Choosing Representatives}
In the next phase we are going to further reduce the size $N$ of $C$ via the following approach: We first choose a set $C' = \{c_1,\ldots,c_{n'}\} \subset C$ of $n' = \Theta(\sqrt{n})$ candidates uniformly at random.
To do this the anchor propagates $n$ and $N$ to all nodes via an aggregation phase.
Then each node $v$ chooses each of its candidates $w \in v.C$ independently with probability $\sqrt{n}/N$ and aggregates the number $n'$ of chosen candidates to the anchor.
Following this approach, it is easy to see that $n' \in \Theta(\sqrt{n})$ w.h.p. due to Chernoff bounds.

\subsection{Phase 2b: Distributed Sorting}
Our next goal is to compute the \emph{order} of each candidate in $C'$ when sorting them via their priorities (see \Cref{algo:sorting}): For this we let the  anchor assign a unique position $pos(c_i) \in \{1,\ldots,n'\}$ to each candidate $c_i \in C'$ via decomposition of the interval $[1,n']$ over the aggregation tree (similar to Phase~$3$ of \pname{}).

\begin{algorithm}[ht]
\caption{Distributed Sorting}
\label{algo:sorting}
\begin{algorithmic}[1]
	\Statex \textbf{Input}: $c_1,\ldots,c_{n'} \in C'$
	\Statex \textbf{Output}: Order for each $c_1,\ldots,c_{n'}$ based on priorities
	\Statex 
	\Statex \textbf{Algorithm} (executed for each $c_i$)
	\State \pushcode[0] Assign a unique position $pos(c_i) \in \{1,\ldots,n'\}$ to $c_i$
	\State \pushcode[0] Route $c_i$ to the node $v_i \in V$ responsible for $pos(c_i)$
	\State \pushcode[0] Distribute $n'$ copies $c_{i,1},\ldots,c_{i,n'}$ of $c_i$ over
	\Statex \pushcode[0] $v_{i,1},\ldots,v_{i,n'} \in T(v_i)$
	\State \pushcode[0] Route copy $c_{i,j}$ to $w_{i,j} \in V$ responsible for $h(i,j)$
	\State \pushcode[0] \textbf{if} $prio(c_{i,j}) > prio(c_{j,i})$ \textbf{then} \Comment{$c_{i,j}$ and $c_{j,i}$ meet at $w_{i,j}$}
	\State \pushcode[0] \ \ \ \ Send $(1,0)$ to $v_{i,j}$, send $(0,1)$ to $v_{j,i}$
	\State \pushcode[0] \textbf{else}
	\State \pushcode[0] \ \ \ \ Send $(0,1)$ to $v_{i,j}$, send $(1,0)$ to $v_{j,i}$
	\State \pushcode[0] Aggregate \& combine vectors to the root node $v_i \in T(v_i)$
	\Statex \pushcode[0] to obtain the order of $c_i$
\end{algorithmic}
\end{algorithm}

Each node routes each of its chosen $c_i \in C'$ to the node $v_i$ responsible for position $pos(c_i)$ (similar to Phase~$4$ of \pname{}).
Then each node $v_i$ generates $n'$ copies of $c_i$ and distributes them to $n'$ other nodes in the following way: Let $b(v_i) = (v_{i,1},\ldots,v_{i,d})$ be the first $d=\log n'$ bits of $v_i$'s unique bitstring according to the classical de Bruijn graph (recall that the aggregation tree is able to emulate routing in the classical de Bruijn graph due to \Cref{lemma:aggregation_tree}(v)).
Then $v_i$ stores a pair $([n'/2,n'/2],c_i)$ for itself and sends a pair $([1,n'/2-1],c_i)$ to the node with bitstring $(0,v_{i,1},\ldots,v_{i,d-1})$ and another pair $([n'/2+1,n'],c_i)$ to the node with bitstring $(1,v_{i,1},\ldots,v_{i,d-1})$.
Repeating this process recursively until a node receives a pair $([a,b],c_i)$ with $a = b$ guarantees that $n'$ nodes now hold a copy of $c_i$.
Observe that this approach induces a (unique) aggregation tree $T(v_i)$ with root $v_i$ and height at most $\Theta(\log \sqrt{n})$, when nodes remember the sender on receipt of a copy of $c_i$.
Furthermore, there is no node serving as a bottleneck, i.e., the number of aggregation trees that a node participates in is only constant on expectation:

\begin{lemma} \label{lemma:aggregation_tree:congestion}
	Let $T(v_1),\ldots,T(v_{n'})$ be the unique aggregation trees as defined above.
	Then for all $w \in V$ it holds $E[|\{T(v_i)\ |\ w \in T(v_i)\}|] = \Theta(1)$.
\end{lemma}

\begin{proof}
	Having $N$ remaining candidates $c_1,\ldots,c_N$, there are $N$ unique trees $T(v_1),\ldots,T(v_N)$ out of which we choose $n'$ uniformly at random, i.e., $\Pr[\text{Tree } T(v_i) \text{ selected}] = n'/N = \Theta(\sqrt{n}/N)$.
	As each tree has height at most $\log n'$, the number of nodes in each tree is equal to $\sum_{i=0}^{\log n'} 2^i = 2 \cdot 2^{\log n'} = 2n'$.
	Observe that since we choose the root nodes of each tree uniformly and independently and the tree height is only $\log n'$, all nodes in the tree are determined uniformly and independently.
	Thus, the probability that a node $w$ is part of tree $T$ is equal to $2n'/n = \Theta(1/\sqrt{n})$.
	Therefore we can compute the expected number of trees that $w$ is part of, i.e., $E[|\{T(v_i)\ |\ w \in T(v_i)\}|] = N \cdot \Theta(\frac{\sqrt{n}}{N}) \cdot \Theta(\frac{1}{\sqrt{n}}) = \Theta(1)$.
\end{proof}

Denote the element $c_{i,j}$ as the $j^{\mathit{th}}$ copy of $c_i$, i.e., $c_{i,j}$ is the candidate $c_i$ that is passed as part of the pair $([j,j],c_i)$ previously.
Let $v_{i,j}$ be the node in $T(v_i)$ that received $c_{i,j}$.
Then $v_{i,j}$ uses the pseudorandom hash function $h:\{1,\ldots,n'\}^2 \rightarrow [0,1)$ with $h(i,j) = h(j,i)$ for any $i,j \in \{1,\ldots,n'\}$ to route $c_{i,j}$ to the node in the DHT maintaining the key $h(i,j)$.
Node $v_{i,j}$ also sends a reference to itself along with $c_{i,j}$.
Once we have done this for all copies on all $n'$ aggregation trees, a node $w_{i,j}$ that is responsible for position $h(i,j)$ now has received the following data: The copy $c_{i,j}$ along with the node $v_{i,j}$ and the copy $c_{j,i}$ along with the node $v_{j,i}$.
Thus, $w_{i,j}$ can compare the priorities of $c_{i,j}$ and $c_{j,i}$.
Based on the result of the comparison, $w_{i,j}$ sends a vector $(1,0)$ to $v_{i,j}$ and a vector $(0,1)$ to $v_{j,i}$ (in case $prio(c_{i,j}) > prio(c_{j,i})$) or a vector $(0,1)$ to $v_{i,j}$ and a vector $(1,0)$ to $v_{j,i}$ (in case $prio(c_{i,j}) < prio(c_{j,i})$).
When $v_{i,j}$ receives a vector $(1,0)$ then that means that one of the nodes in $C'$ has a smaller priority than $c_i$.
Next, we aggregate all these vectors to the root of each aggregation tree $T(v_i)$, using standard vector addition for combining.
This results in $v_i$ knowing the order of candidate $c_i$ in $C'$: If the combined vector at $v_i$ is a vector $(L,R) \in \mathbb{N}^2$, then the order of $c_i$ is equal to $L+1$.

\subsection{Phase 2c: Reducing Candidates}
In the next step the anchor computes two candidates $c_l$ and $c_r$, such that we can guarantee w.h.p. that the element of rank $k$ lies between those candidates.
For this we consider the candidate $c_k \in C'$ for which the expected rank is closest to $k$, i.e., the candidate $c_k$ with order $k\frac{n'}{N}$.
Then we move $\delta$ candidates to the left/right in the ordering of candidates in $C'$.
Let $c_l \in C'$ be the candidate whose order is equal to $l = \lfloor k\frac{n'}{N} - \delta \rfloor$ and $c_r \in C'$ be the candidate whose order is equal to $r = \lceil k\frac{n'}{N} + \delta \rceil$.
In case $l < 1$ we just consider $c_r$ and in case $r > n'$ we just consider $c_l$.
For now, we just assume $l \geq 1$ and $r \leq n'$.
We delegate $c_l$ and $c_r$ up to the anchor in the (standard) aggregation tree.

Once the anchor knows $c_l$ and $c_r$, it sends them to all nodes in the aggregation tree.
Now we compute the exact ranks of $c_l$ and $c_r$ in $C$ via another aggregation phase: Each node $v$ computes a vector $(l_v,r_v) \in \mathbb{N}^2$, where $l_v$ represents the number of candidates in $v.C$ with smaller priority than $c_l$ and $r_v$ represents the number of candidates in $v.C$ with smaller priority than $c_r$.
This results in the aggregation of a vector $(L,R) \in \mathbb{N}^2$ when using standard vector addition at each node in the tree.
Once $(L,R)$ has arrived at the anchor, it knows $rank(c_l)$ and $rank(c_r)$.
To finish the iteration, the anchor updates $v_0.k$ to $v_0.k-rank(c_l)$ and tells all nodes $v$ in another aggregation phase to remove all candidates $w \in v.C$ with $rank(w) < rank(c_l)$ or $rank(w) > rank(c_r)$ and aggregate the overall number $k'$ of those candidates up to the anchor, such that it can update $v_0.N$.
Then the anchor starts the next iteration (in case $v_0.N > \sqrt{v_0.n}$) or switches to the last phase ($v_0.N < \sqrt{v_0.n}$).

We now show that this approach further reduces the number of candidates.
First we want to compute the necessary number of shifts $\delta$ such that $rank(c_l) < k$ for $c_l$ and $rank(c_r) > k$ for $c_r$ holds w.h.p., as this impacts the number of candidates that are left for the next iteration of the second phase.
For this we need the following technical lemma:

\begin{lemma} \label{lemma:k_select:phase_2_correctness}
	If $\delta \in \Theta(\sqrt{\log n} \cdot \sqrt[4]{n})$, then w.h.p. $rank(c_l) < k$ and $rank(c_r) > k$.
\end{lemma}

\begin{proof}
	We just show $rank(c_l) < k$, as the proof for $rank(c_r) > k$ works analogously.
	Let $c_k \not \in C'$ be the element with $rank(c_k) = k$.
	Let $X_i = 1$, if the candidate $c_i \in C$ with $rank(c_i) = i$ has been chosen to be in $C'$ in the first Phase~2a.
	Let $X = \sum_{i=1}^k X_i$.
	Then $E[X] = k \sqrt{n}/N$.
	The probability that too few candidates with rank smaller than $k$ have been chosen to be in $C'$ should be negligible, i.e., $\Pr[X \leq E[X] - \delta] \leq n^{-c}$ for some constant $c$, where $\delta$ denotes the number of steps that we have to go to the left from the candidate with order $k \sqrt{n}/N$.
	In order to apply Chernoff bounds, we first compute $\varepsilon > 0$ such that $\Pr[X \leq (1-\varepsilon)E[X]] = \Pr[X \leq E[X] - \delta]$: Solving the equation $(1-\varepsilon)E[X] = E[X] - \delta$ for $\varepsilon$ yields $\varepsilon = \delta/E[X]$.
	Using Chernoff bounds on $\Pr[X \leq (1-\varepsilon)E[X]]$ results in \[\Pr[X \leq (1-\varepsilon)E[X]] \leq \exp(-\varepsilon^2 E[X]/2) \leq n^{-c}\] for $\varepsilon^2 E[X] \geq c \log n$, $c > 0$ constant.
	Solving this equation for $\varepsilon$ leads to $\varepsilon \geq \sqrt{\frac{c \log n}{E[X]}}$.
	By solving the equation \[\frac{\delta}{E[X]} \geq \sqrt{\frac{c \log n}{E[X]}}\] for $\delta$ we get \[\delta \geq \sqrt{c \log n \cdot E[X]} \overset{E[X] = \sqrt{n}}{=} \Theta(\sqrt{\log n} \cdot \sqrt[4]{n}).\]
\end{proof}

\begin{lemma} \label{lemma:k_select:phase_2_reduction}
	After $\Theta(1)$ iterations of the second phase, $N \in O(\sqrt{n})$ w.h.p.
\end{lemma}

\begin{proof}
	Recall that by \Cref{lemma:k_select:first_iteration}, we have $N = O(n^{3/2} \cdot \log n)$ after the first phase.
	Consider the candidates $c_l$ and $c_r$ as determined by the anchor.
	Due to \Cref{lemma:k_select:phase_2_correctness} it holds $rank(c_l) < k < rank(c_r)$ and there are $\delta \in \Theta(\sqrt{\log n} \cdot \sqrt[4]{n})$ candidates lying between $c_l$ and $c_r$ that are contained in $C'$, i.e., we consider the ordered sequence $c_l, c_{l+1},\ldots,c_{r-1},c_r$ of candidates in $C'$.
	We compute the number $\beta$ of candidates that lie between two consecutive candidates $c_i, c_{i+1} \in C'$ such that the probability that all $\beta$ candidates have not been chosen in Phase~2a becomes negligible, yielding an upper bound for the number of candidates lying between $c_i$ and $c_{i+1}$.
	Recall that the probability that a candidate is chosen to be in $C'$ is equal to $\sqrt{n}/N$.
	We get $\Pr[\beta\ \text{ candidates between}\ c_i\ \text{and}\ c_{i+1}\ \text{are not chosen}] = (1 - \sqrt{n}/N)^{\beta} \leq \exp(-\frac{\sqrt{n}}{N} \cdot \beta) = n^{-c}$ for $\beta = c \cdot \frac{N}{\sqrt{n}} \cdot \ln n = N \cdot \Theta(\frac{\log n}{\sqrt{n}})$ and a constant $c$.
	Overall, it follows that $N$ is reduced by factor $\Theta(\delta\frac{\log n}{\sqrt{n}})$ in each iteration of the second phase w.h.p.
	After five iterations of the second phase $N$ is reduced to 
	\begin{eqnarray*}
		N\cdot \left(\delta \frac{\log n}{\sqrt{n}}\right)^5 & \overset{\Cref{lemma:k_select:first_iteration}}{=} & n^{3/2} \cdot \log n \cdot \left(\delta \frac{\log n}{\sqrt{n}}\right)^5 \\
		& = & \log^{8}(n) \cdot \sqrt{\log n} \cdot \sqrt[4]{n} \\
		& = & O(\sqrt{n}).
	\end{eqnarray*}	
\end{proof}

Note that in case $l < 1$ (analogously $r > n'$), the set $\{c_1,\ldots,c_r\} \subset C'$ contains at most $\delta \in \Theta(\sqrt{\log n} \cdot \sqrt[4]{n})$ candidates, so \Cref{lemma:k_select:phase_2_reduction} still holds.

\subsection{Phase 3: Exact Computation}
The third and last phase computes the exact $k^{\mathit{th}}$ smallest element out of the remaining candidates.
This phase is basically just a single iteration of the second phase with the exception that each remaining candidate is now chosen to be in $C'$ in Phase~2a, leading to each candidate being compared with each other remaining candidate.
This immediately gives us the exact rank of each remaining candidate, as it is now equal to the determined order, so we are able to send the candidate that is the $k^{\mathit{th}}$ smallest element to the anchor.

We are now ready to show \Cref{theorem:k_selection:runtime}:

\begin{proof}[Proof of \Cref{theorem:k_selection:runtime}]
	It is easy to see that in all three phases we perform a constant amount of aggregation phases for a constant amount of iterations.
	Note that in the second and third phase the time for a DHT-insert is $O(\log n)$ w.h.p. due to \Cref{lemma:aggregation_tree}(iii).
	Also note that we perform the actions that have to be done in each of the generated $n' = \Theta(\sqrt{n})$ aggregation trees in parallel, resulting in a logarithmic number of rounds until the order of each chosen candidate is determined.
	As a single aggregation phase takes $O(\log n)$ rounds w.h.p., we end up with an overall running time of $O(\log n)$ w.h.p. for \textsc{KSelect}.
	
	For the congestion bound, note that the only time we generate more than a constant amount of congestion at nodes is in the second phase when routing the chosen candidates $c_i \in C'$ to the node $v_i$ responsible for $pos(c_i)$ in $O(\log n)$ rounds w.h.p. 
	Thus, as each node chooses $\frac{\sqrt{n}}{N} \cdot \frac{N}{n} = \frac{\sqrt{n}}{n} = O(1)$ of its candidates to be in $C'$ on expectation, one can easily verify via Chernoff bounds and \Cref{lemma:aggregation_tree}(v) that this generates a congestion of $\widetilde{O}(1)$ w.h.p.
	With the same argumentation in mind, observe that the congestion generated for nodes that are part of at least one aggregation tree $T(v_i)$ is constant w.h.p., because each node participates in only two such aggregation trees on expectation (\Cref{lemma:aggregation_tree:congestion}).
	Participation of node $v$ in one of these aggregation trees means that $v$ has to perform only one single comparison of priorities, leaving the congestion constant.
	
	Finally, one can easily see that the size of each message is $O(\log n)$ bits, because messages in \textsc{KSelect} contain only a constant amount of elements, where each element can be encoded by $O(\log n)$ bits due to its priority being within $\{1,\ldots,n^q\}$.
\end{proof}

\section{Arbitrary Priorities} \label{sec:heap:prio2}
We are now ready to demonstrate how to use the protocol \textsc{KSelect} from the previous section in order to realize \pnamep{}.
\pnamep{} is able to support larger amounts of priorities, i.e., $\mathcal P = \{1,\ldots,n^q\}$ for some constant $q \in \mathbb{N}$.
In order to provide a scalable solution, we give up on the local consistency semantic (\Cref{def:semantics}), which makes \pnamep{} serializable instead of sequentially consistent.

The general idea for processing operations in \pnamep{} is roughly the same as in \pname{}: We first aggregate batches in the aggregation tree to the anchor, but instead of a batch representing both \ins{} and \delmin{} requests, we only aggregate the overall number of \ins{} requests or the overall number of \delmin{} requests.
Consequently, we distinguish between a separate \ins{} phase and a \delmin{} phase.
Whenever a node $v$ generates a new \ins{} or \delmin{} request it stores the request in a local queue that acts as a buffer.
\Cref{algo:seap} summarizes \pnamep{}.

\begin{algorithm}[ht]
\caption{Protocol \pnamep}
\label{algo:seap}
\begin{algorithmic}[1]
	\Statex \textbf{Insert Phase}
	\State \pushcode[0] Aggregate the number $k \in \mathbb{N}$ of insertions to the anchor
	\State \pushcode[0] $v_0.m \gets v_0.m + k$
	\State \pushcode[0] Broadcast start of insertions over the tree
	\State \pushcode[0] Store elements at random nodes
	\Statex
	\Statex \textbf{DeleteMin Phase}
	\State \pushcode[0] Aggregate the number $k \in \mathbb{N}$ of deletions to the anchor  \label{algo:seap:line_5}
	\State \pushcode[0] Determine the element with rank $k$ using \textsc{KSelect}  \label{algo:seap:line_6}
	\State \pushcode[0] Assign a unique position $pos \in \{1,\ldots,k\}$ to the $k$ most 
	\Statex \pushcode[0] prioritized elements \label{algo:seap:line_7}
	\State \pushcode[0] Store these elements at the node maintaining the key $h(pos)$\label{algo:seap:line_8}
	\State \pushcode[0] Assign a unique sub-interval $[a,b] \subset [1,k]$ to each node
	\Statex \pushcode[0] that has to execute $b-a+1$ \delmin{} operations \label{algo:seap:line_9}
	\State \pushcode[0] Fetch the elements stored at positions $\{a,\ldots,b\}$  \label{algo:seap:line_10}
\end{algorithmic}
\end{algorithm}

\subsection{Insert Phase}
At the beginning of the \ins{} phase of \pnamep{} each node $v$ generates a snapshot of the number of \ins{} operations stored in its queue and stores it in a variable $v.I$.
Similar to \pname{}, nodes aggregate all $v.I$'s to the anchor, using simple addition to combine two numbers $v.I$ and $v'.I$.
When the anchor $v_0$ receives the aggregated value $v_0.I^{+}$ at the end of the first phase, it updates $v_0.m$ and announces over the aggregation tree that nodes are now allowed to process \textsc{Put} operations on the DHT.
For each element $e$ that node $v$ wants to store in the DHT it assigns a key $key(e) \in \mathbb{N}$ generated uniformly at random and sends $e$ to the node $w$ that is responsible for $key(e)$.
Once $w$ has received $e$, it sends a confirmation message back to $v$.
Upon receiving all confirmations for all its elements $v$ switches to the \delmin{} phase.

\subsection{DeleteMin Phase} \label{sec:heap:prio2:del_min}
Aggregation of \delmin{} requests works analogously to the \ins{} phase.
At the end of the aggregation, the anchor $v_0$ receives a value $k$ representing the number of \delmin{} requests to be processed.
Now we use \textsc{KSelect} to find the element with rank $k$.
In order to assign a unique position $pos \in \{1,\ldots,k\}$ to the $k$ most prioritized elements, we proceed analogously as in Phase~$3$ of \pname{} by decomposing the interval $[1,k]$ into sub-intervals.
Each node $v$ assigns such a position to all its stored elements which have a rank less than $k$.
This can be determined by sending the priority of the $k^{\mathit{th}}$ smallest element along with  each sub-interval.

The decomposition approach from Phase~$3$ of \pname{} is also used to assign a unique sub-interval $[a,b] \subset [1,k]$ to each node that wants to execute $b-a+1$ \delmin{} operations.
For the last step of our algorithm consider a node $v$ that wants to issue $d$ \delmin{} operations on the heap and consequently got assigned to the sub-interval $[a,b]$, $d = b - a + 1$.
Then $v$ generates a \textsc{Get}$(h(pos), v)$ request for each position $pos \in \{a,\ldots,b\}$ to fetch the element that previously got stored in the DHT at that specific position.
This way, each \delmin{} request got a value returned so nodes can then proceed with an \ins{} phase afterwards.

\subsection{Results}
We show the following theorem in this section using a series of lemmas.
The theorem summarizes our results for \pnamep{}:

\begin{theorem} \label{theorem:seap}
	\pnamep{} implements a distributed heap with the following properties:
	\begin{enumerate}
		\item \pnamep{} is fair.
		\item \pnamep{} satisfies serializability and heap consistency.
		\item The number of rounds, needed to process heap requests successfully is $O(\log n)$ w.h.p.
		\item The congestion of \pnamep{} is at most $\widetilde{O}(\Lambda)$.
		\item Messages in \pnamep{} have size of at most $O(\log n)$ bits.
	\end{enumerate}
\end{theorem}

Fairness is clear, since we make use of a pseudorandom hash function in order to distribute heap elements uniformly over all nodes.

\begin{lemma}
	\pnamep{} satisfies serializability and heap consistency.
\end{lemma}

\begin{proof}
	To show serializability, we define the total order $\prec$ for all \ins{} and \delmin{} requests, whose serial execution is equivalent to the distributed execution of requests in \pnamep{}.
	Let $S$ be the set of all \ins{} and \delmin{} requests to be issued on the heap.
	We split $S$ into subsets $S_I$ and $S_D$ with $S_I \cup S_D = S$ and $S_I \cap S_D = \emptyset$.
	$S_I$ contains all \ins{} requests and $S_D$ contains all \delmin{} requests.
	For subset $S_I$ we fix a randomly chosen permutation of the operations, i.e., $S_I = (\textsc{Ins}_1,\ldots,\textsc{Ins}_k)$.
	For subset $S_D$ note that all requests in $S_D$ are issued in the same \delmin{} phase and thus each request $\textsc{Del}_i$ is assigned a unique position $pos(\textsc{Del}_i)$ by the anchor, such that we can remove the heap element that is stored at that position in the DHT.
	Then we choose the permutation $S_D = (\textsc{Del}_1,\ldots,\textsc{Del}_l)$ such that $pos(\textsc{Del}_i) = i < i + 1 = pos(\textsc{Del}_{i+1})$ for all $i \in \{1,\ldots,l-1\}$.
	We define $\prec$ to be the total order $T$ induced by the combination of the above permutations, i.e., $T = (\textsc{Ins}_1,\ldots,\textsc{Ins}_k,\textsc{Del}_1,\ldots,\textsc{Del}_l)$.
	This is in accordance with our protocol, because we handle \ins{} and \delmin{} requests in separate phases, i.e., we wait until all \ins{} requests have been processed before we start processing all \delmin{} requests.
	Hence, \pnamep{} satisfies serializability.
	
	Thus, all that is left is to show that \pnamep{} satisfies heap consistency, so we show the properties of \Cref{def:heap_consistency} for $\prec$ in the following.
	
	\begin{itemize}
		\item[(1)] Let $(\textsc{Ins}_i, \textsc{Del}_j) \in M$.
		Then by definition of $\prec$ it follows $\textsc{Ins}_i \prec \textsc{Del}_j$.
		\item[(2)] Let $(\textsc{Ins}_i, \textsc{Del}_j) \in M$ and assume that there exists $\textsc{Del}_k \not \in M$ such that $\textsc{Ins}_i \prec \textsc{Del}_k \prec \textsc{Del}_j$.
		By definition of $\prec$ it holds $pos(\textsc{Del}_k) < pos(\textsc{Del}_j)$ and thus, the size of the heap is at most $pos(\textsc{Del}_k)$ as $pos(\textsc{Del}_k)$ did not get matched to a corresponding \ins{} operation.
		It follows that $\textsc{Del}_j$ should not have been matched either as its position is greater than the size of the heap, which is a contradiction.
		\item[(3)] Let $(\textsc{Ins}_i, \textsc{Del}_j) \in M$ and assume that there exists $\textsc{Ins}_k \not \in M$ such that $\textsc{Ins}_k \prec \textsc{Del}_j$ and $prio(\textsc{Ins}_k) < prio(\textsc{Ins}_i)$.
		It follows that $\textsc{Ins}_k$ should have gotten assigned to a position in the \delmin{} phase such that $pos(\textsc{Ins}_k) < pos(\textsc{Ins}_i) \overset{!}{=} pos(\textsc{Del}_j)$.
		The way we defined $S_D$ implies that there exists a request $\textsc{Del}_i$ with $pos(\textsc{Del}_i) = pos(\textsc{Ins}_k)$.
		Thus, $\textsc{Ins}_k$ should have been matched, which is a contradiction to our assumption that $\textsc{Ins}_k \not \in M$.
	\end{itemize}
	This concludes the proof of the lemma.
\end{proof}

The next lemma serves as a proof for the number of rounds needed to process heap requests successfully:

\begin{lemma} \label{theorem:leap:runtime}
	The \ins{} phase and the \delmin{} phase of \pnamep{} finish after $O(\log n)$ rounds w.h.p.
\end{lemma}

\begin{proof}
	The runtime of the \ins{} phase follows from \Cref{lemma:aggregation_tree}(i) and (iii).
	We argue that each step in the \delmin{} phase listed in \Cref{algo:seap} takes at most $O(\log n)$ rounds w.h.p.
	\Cref{algo:seap:line_5},~\Cref{algo:seap:line_7} and~\Cref{algo:seap:line_9} involve an aggregation phase via the aggregation tree.
	Such a phase can be done in $O(\log n)$ rounds w.h.p.
	\Cref{algo:seap:line_6} terminates after $O(\log n)$ rounds w.h.p. according to \Cref{theorem:k_selection:runtime}.
	Finally, inserting elements simultaneously into the DHT in \Cref{algo:seap:line_8} takes $O(\log n)$ rounds w.h.p. according to \Cref{lemma:aggregation_tree}(iii).
	Using the same argumentation, \Cref{algo:seap:line_10} can also be processed in $O(\log n)$ rounds w.h.p.
\end{proof}

\begin{lemma}
	The congestion of \pnamep{} is at most $\widetilde{O}(\Lambda)$.
\end{lemma}

\begin{proof}
	At the beginning an \ins{} phase, each node $v$ has at most $\lambda(v) \cdot O(\log n) = \widetilde{O}(\lambda(v))$ \ins{} requests buffered, since the previous phase lasted for $O(\log n)$ rounds (\Cref{theorem:leap:runtime}) and $v$ could have generated at most $\lambda(v)$ requests per round.
	For each of those requests, $v$ delegates the element to a randomly chosen node, resulting in $v$ having to process $\widetilde{O}(\lambda(v))$ requests at once.
	Since each of these delegations needs $O(\log n)$ rounds w.h.p. (\Cref{lemma:aggregation_tree}(iii)) to finish and the aggregation tree only generates congestion up to a polylogarithmic factor (\Cref{lemma:aggregation_tree}(v)), the lemma follows for the \ins{} phase.

	For the \delmin{} phase, \Cref{algo:seap:line_5},~\Cref{algo:seap:line_7} and~\Cref{algo:seap:line_9} of \Cref{algo:seap} can be processed via aggregation phases and thus have no impact on the upper bound for the congestion.
	For \Cref{algo:seap:line_6} the congestion is $\widetilde{O}(1)$ due to \Cref{theorem:k_selection:runtime}.
	For the last step at which nodes fetch data from the DHT (\Cref{algo:seap:line_10}) we can use the same argumentation as for the \ins{} phase, so the lemma follows.
\end{proof}

Finally we show the bound on the size of messages for \pnamep{}:

\begin{lemma}
	Messages in \pnamep{} have size of at most $O(\log n)$ bits.
\end{lemma}

\begin{proof}
	First note that the number of \ins{} and \delmin{} requests can only be polynomial in $n$: Due to \Cref{theorem:leap:runtime} each node may generate up to $\Lambda \in O(poly(n))$ new heap requests for $\log n$ rounds until the next phase is started.
	Thus, aggregating the number of \ins{} or \delmin{} requests to the anchor yields a message size of $O(\log n)$.
	Note that in order to realize \Cref{algo:seap:line_7} and~\Cref{algo:seap:line_9} of \Cref{algo:seap} we only need to store a single interval in each message (similar to Phase 3 of \pname{}).
	Keeping this argumentation in mind and the fact that \textsc{KSelect} uses only $O(\log n)$ bit messages (\Cref{theorem:k_selection:runtime}), one can easily see that \pnamep{} uses only $O(\log n)$ bit messages.
\end{proof}

\section{Conclusion} \label{sec:conclusion}
We presented two protocols \pname{} and \pnamep{} to realize distributed heaps along with a novel protocol \textsc{KSelect} that solves the distributed $k$-selection problem.

One may ask whether we could realize \delmin{} in $O(1)$ time like in centralized heaps, while preserving scalability.
This would mean that we have to be able to access elements in constant time from any process in the system, implying that the degree for processes has to go up to at least $\Omega(\sqrt{n})$.
As a consequence, the update costs for \join{} and \leave{} operations would rise drastically.

Also, can we modify \pnamep{} in order to also guarantee sequential consistency, i.e., how to realize local consistency in \pnamep{}?
A first idea would be to maintain the same batches as in \pname{}, but only aggregate the first amount of \ins{} or \delmin{} operations to the anchor.
This suffices to provide sequential consistency for \pnamep{} but comes at the cost of scalability and message size, as batches maintained at nodes may grow infinitely long for high injection rates.

%%
%% Bibliography
%%

%% Please use bibtex, 
\bibliographystyle{alpha}
\bibliography{literature}

\newpage
\appendix

\section{Aggregation Tree Construction} \label{appendix:aggregation_tree}

In order to define a tree with the properties states in \Cref{lemma:aggregation_tree}, we adapt a dynamic version of the de Bruijn graph from~\cite{DBLP:conf/sss/RichaSS11}, which is based on~\cite{DBLP:journals/talg/NaorW07}, for our network topology:

\begin{definition}
	The \emph{Linearized de Bruijn network} (LDB) is a directed graph $G = (V, E)$, where each node $v$ emulates $3$ (virtual) nodes: A \emph{left virtual node} $l(v) \in V$, a \emph{middle virtual node} $m(v) \in V$ and a \emph{right virtual node} $r(v) \in V$.
	The middle virtual node $m(v)$ has a real-valued \emph{label}\footnote{We may indistinctively use $v$ to denote a node or its label, when clear from the context.} in the interval $[0,1)$.
	The label of $l(v)$ is defined as $m(v)/2$ and the label of $r(v)$ is defined as $(m(v)+1)/2$.
	The collection of all virtual nodes $v \in V$ is arranged in a sorted cycle ordered by node labels, and $(v,w) \in E$ if and only if $v$ and $w$ are consecutive in this ordering (\emph{linear edges}) or $v$ and $w$ are emulated by the same node (\emph{virtual edges}).
\end{definition}

We assume that the label of a middle node $m(v)$ is determined by applying a publicly known pseudorandom hash function on the identifier $v.id$.
We say that a node $v$ is \emph{right} (resp. \emph{left}) of a node $w$ if the label of $v$ is greater (resp. smaller) than the label of $w$, i.e., $v > w$ (resp. $v < w$). 
If $v$ and $w$ are consecutive in the linear ordering and $v < w$ (resp. $v > w$), we say that $w$ is $v$'s \emph{successor} (resp. \emph{predecessor}) and denote it by $succ(v)$ (resp. $pred(v)$).
As a special case we define $pred(v_{\mathit{min}}) = v_{\mathit{max}}$ and $succ(v_{\mathit{max}}) = v_{\mathit{min}}$, where $v_{\mathit{min}}$ is the node with minimal label value and $v_{\mathit{max}}$ is the node with maximal label value.
This guarantees that each node has a well defined predecessor and successor on the sorted cycle.
More precisely, each node $v$ maintains two variables $pred(v)$ and $succ(v)$ for storing its predecessor and successor nodes.
We assume that $v$ knows whether $pred(v)$ and $succ(v)$ contains a left, middle or right virtual node.
By adopting the result from~\cite{DBLP:conf/sss/RichaSS11}, one can show that routing in the LDB can be done in $O(\log n)$ rounds w.h.p.:

\begin{lemma}\label{lemma:LDB:routing}
	For any $p \in [0,1)$, routing a message from a source node $v$ to a node that is the predecessor of $p$ (i.e., the node closest from below to $p$) in the LDB can be done in $O(\log n)$ rounds w.h.p.
\end{lemma}

The routing strategy in the LDB emulates routing in the $d$-dimensional de Bruijn graph with $d \approx \log n$, and its analysis shows that the message only falls behind by at most $O(\log n)$ hops compared to its position in the $d$-dimensional de Bruijn graph, w.h.p. 
When taking this and the fact into account that a node in the LDB might be responsible for the emulation of up to $O(\log n)$ nodes in the $d$-dimensional de Bruijn graph,  w.h.p., one immediately obtains the following lemma.

\begin{lemma}\label{lemma:LDB:congestion}
	Any routing schedule with dilation $D$ and congestion $C$ in the $d$-dimensional de Bruijn graph can be emulated by the LDB with $n \geq 2$ nodes with dilation $O(D + \log n)$ and congestion $\widetilde{O}(C)$ w.h.p.	
\end{lemma}

Next we present how the LDB can support DHT requests \textsc{Put}($k$, $e$) and \textsc{Get}($k$, $v$).
\textsc{Put}($k$, $e$) stores the element $e$ at the (virtual) node $v$, with $v \leq k < succ(v)$.
\textsc{Get}($k$, $v$) searches for the node $w$ storing the element $e$ with key $k$.
At $w$ we remove $e$ and deliver it back to $v$.
By \Cref{lemma:LDB:routing}, we can support \textsc{Put} and \textsc{Get} requests in $O(\log n)$ hops w.h.p.

The LDB consisting of all virtual nodes contains the aggregation tree as a subgraph.
The parent node $p(v)$ of a node $v$ is defined as follows: If $v$ is a middle virtual node, then $p(v) = l(v)$.
If $v$ is a left virtual node then $p(v) = pred(v)$.
Finally, if $v$ is a right virtual node, then $p(v) = m(v)$.
Next, we describe how a node $v$ knows its child nodes (denoted by the set $C(v)$) in the aggregation tree.
If $v$ is a middle virtual node, then either $C(v) = \{r(v), succ(v)\}$ (if $succ(v)$ is a left virtual node) or $C(v) = \{r(v)\}$ (otherwise).
If $v$ is a left virtual node, then either $C(v) = \{m(v), succ(v)\}$ (if $succ(v)$ is a left virtual node) or $C(v) = \{m(v)\}$ (otherwise).
Last, if $v$ is a right virtual node, then $C(v) = \emptyset$.
Observe that each node $v$ is able to locally detect its parent and its child nodes in the aggregation tree depending on whether $v$ is a left or middle virtual node (see \Cref{fig:aggregation_tree} for an example).

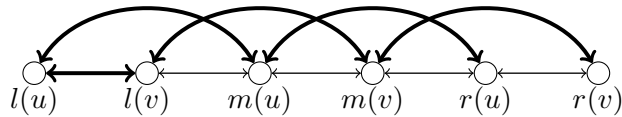
\begin{figure}[ht]
 	\centering
 	\begin{tikzpicture}[main node/.style={circle,draw,align=center, minimum size=0.3cm, inner sep=0pt}]	
     	\node[main node, label={[yshift=-0.85cm]$l(u)$}] (A) at (0,0) {};
     	\node[main node, label={[yshift=-0.85cm]$l(v)$}] (B) at (1.5,0) {};
     	\node[main node, label={[yshift=-0.85cm]$m(u)$}] (C) at (3,0) {};
     	\node[main node, label={[yshift=-0.85cm]$m(v)$}] (D) at (4.5,0) {};
     	\node[main node, label={[yshift=-0.85cm]$r(u)$}] (E) at (6,0) {};
     	\node[main node, label={[yshift=-0.85cm]$r(v)$}] (F) at (7.5,0) {};
     		
     	\draw[<->, line width=1.5pt] (A) to (B);
     	\draw[<->, line width=0.5pt] (B) to (C);
     	\draw[<->, line width=0.5pt] (C) to (D);
     	\draw[<->, line width=0.5pt] (D) to (E);
     	\draw[<->, line width=0.5pt] (E) to (F);
     	\draw[<->, line width=1.5pt, bend left = 60] (A) to (C);
     	\draw[<->, line width=1.5pt, bend left = 60] (B) to (D);
     	\draw[<->, line width=1.5pt, bend left = 60] (C) to (E);
     	\draw[<->, line width=1.5pt, bend left = 60] (D) to (F);
 	\end{tikzpicture}
 	\caption{A LDB consisting of $6$ virtual nodes (corresponding to $2$ nodes $u,v \in V$). Bold linear/virtual edges define the corresponding aggregation tree.}
 	\label{fig:aggregation_tree}
\end{figure}

From \Cref{lemma:LDB:routing}, we directly obtain the upper bound for the height of the aggregation tree:

\begin{corollary} \label{cor:tree_height}
	The aggregation tree based on the LDB has height $O(\log n)$ w.h.p.
\end{corollary}

\end{document}